\newcommand{\argmax}{\text{argmax}}
\newcommand{\bX}{{\bf X}}
\newcommand{\bY}{{\bf Y}}
\newcommand{\hbX}{\hat{\bf X}}
\newcommand{\hbY}{\hat{\bf Y}}
\newcommand{\Cov}{{\bf Cov}}
\newcommand{\subjectto}{\text{ subject to }}

\documentclass[twoside,leqno,twocolumn]{article}  
\usepackage{ltexpprt} 
\usepackage{amssymb}
\usepackage{amsmath}
\usepackage{graphicx}

\begin{document}

\title{\Large Feature selection for high-dimensional integrated data}
\author{Charles Zheng\thanks{Texas A \& M Dept. Statistics}, Scott Schwartz$^*$, Robert S. Chapkin\thanks{Texas A \&M Program in Integrative Nutrition \& Complex Diseases, Center for Environmental \& Rural Health}, \\ Raymond J. Carroll$^*$, Ivan Ivanov\thanks{Texas A \& M Dept. Veterinary Physiology and Pharmacology. All emails should be directed to ivanzau at gmail.com}}
\date{}

\maketitle


\begin{abstract} \small\baselineskip=9pt 
Motivated by the problem of identifying correlations between genes or features of two related biological systems,
we propose a model of \emph{feature selection} in which only a subset of the
predictors $X_t$ are dependent on the multidimensional variate $Y$, and the remainder of the predictors
constitute a ``noise set'' $X_u$ independent of $Y$.
Using Monte Carlo simulations, we investigated the relative performance of two methods: thresholding
and singular-value decomposition, in combination with stochastic optimization to determine
``empirical bounds'' on the small-sample accuracy of an asymptotic approximation.
We demonstrate utility of the thresholding and SVD feature selection methods to with respect to a recent 
infant intestinal gene expression and metagenomics dataset.
\end{abstract}

\section{Introduction.}
\subsection{Motivation.}Our study is motivated by the challenge of performing an integrative analysis of a recent infant intestinal host-metabiome dataset \cite{SCHWARTZ}.
The data consists of microarray intensities for $p=585$ genes, $\bX$,
and next-gen sequencing hits for microbial DNA fragments organized into $q=211$ subsystem classes,
collected from stool samples of $n = 6$ newborn babies.
Standard tests reveal conclusive evidence that the gene expression data and microbiome attributes are dependent \cite{SCHWARTZ}.
The next objective  is to qualify the detailed nature of this association; however,
the high dimensionality of the data poses a computational difficulty for modelling.
In order to reduce the dimensionality of the data for initial exploratory modelling,
it is necessary to employ \emph{feature selection} to select a smaller subset of the genes.

\subsection{Background.}
Feature selection in the context of a univariate response has been extensively studied in the
statistics and data mining literature \cite{ESL}.
However, much less has been done on feature selection for a multivariate response vector.
Group lasso \cite{GRLASSO} has been studied as a feature selection method for multivariate linear regression,
but has been generally used for multi-task learning.
Sparse canonical correlation analysis \cite{PARKHOMENKO}\cite{WITTEN2}\cite{WIESEL} has been proposed
specially for high-throughput biological data.
However, sparse CCA does not directly produce a ranking of the features, but rather
returns a list of genes of varying cardinality depending on tuning parameters.
Meanwhile, a factor-analysis-based model \cite{RAI} has been introduced as a bayesian version of
canonical correlation analysis; however, the dimensionality of our data makes bayesian computation impractical.
Therefore, in this paper, we study a simplifed version of sparse CCA which produces a ranking of the features,
which we call the SVD method.

\subsection{Objectives}
The objectives of this current work are to develop tools for investigating
of the performance of two feature selection methods (thresholding and SVD),
and then to apply these tools to inform a integrative 
In section \S\ref{methods} we propose a model for evaluating the performance of the feature selection methods,
develop asymptotic tools for deriving analytical results, and investigate the effectiveness of the asymptotic
approximations using simulation.
In section \S\ref{application} we apply the thresholding and SVD methods to two sets of integrated
microarray-metagenomics data, and use simulation results based on our proposed model to obtain required sampl size
estimates for follow-up experiments.
Further applications of the present body of work are discussed in \S\ref{conclusions}.
 
\section{Methods and Technical Solutions.}\label{methods}
\subsection{Feature Selection Model}
In our application, we hypothesized that associations between the host genes
and bacteria gene expression levels are generally negligible,
except for a small fraction of host genes and microbial gene categories with significant interaction.
Therefore, in our model, we assume that the
host genes with expression levels correlated with the expression levels of the microbial attributes form a small subset $X_t$ of the
host genes $X$, and that the rest of the host genes $X_u$ are independent of the microbial attributes $Y$.
It then follows that, letting $X=(X_t,X_u)$ without loss of generality,
and also putting $Y=(Y_t,Y_u)$ where $Y_u$ is independent of $X$, we have
\begin{equation}
\Sigma_{XY}=\begin{pmatrix}\Sigma_{X_tY_t}&0 \\0 & 0\end{pmatrix}
\end{equation}
where $\Sigma_{X_tY_t} = \Cov(\tilde{X}_t,\tilde{Y_t})$, and where $\tilde{X}_t=\{f_1(X_1),\hdots, f_{p_t}(X_{p_t})\}$,
where $p_t$ is the dimension of $X_t$.

Further assuming that $Cov(X)=I_p$, $Cov(Y)=I_q$, it follows that the covariance matrix of $(X,Y)$ is
\begin{equation}\label{sigma}
\Sigma = \begin{pmatrix}I_{p_t} & 0 & \Sigma_{X_tY_t} & 0 \\
0 & I_{p_u} & 0 & 0\\
 \Sigma_{X_tY_t}^T & 0 & I_{q_t} & 0\\
0 & 0 & 0 & I_{q_u}\end{pmatrix}.
\end{equation}

Now we consider \emph{feature selection algorithms} which return a ranking $\psi: \{1,\hdots,p\} \to \{1,\hdots,p\}$
of the features in $X$.
Here the ranking $\psi$ is formally represented by a bijective map which associates to each ordinal rank $1,\hdots,p$
an index $1,\hdots,p$ of $X$ (thus, we ignore the possibility of ties.)
Thus the feature $X_{\psi(1)}$ is interpreted as the ``most promising feature.''
To formally evaluate ranking methods we use the 1-0 loss for the top-ranked feature:
\begin{equation}
L(\psi) = I(\psi(1) > p_t),
\end{equation}
recalling that we arrange $X$ as $(X_t,X_u)$ so that $X_1,\hdots,X_{p_t}$ are correlated with $Y$.
The rankings $\psi$ can be obtained from real-valued \emph{scores} $s: \{1,\hdots,p\} \to \mathbb{R}$
by letting
\begin{equation}\psi^{-1}(i) = \sum_{j=1}^p I(s(j) > s(i)) + \sum_{j=i}^p I(s(j)=s(i)) \end{equation}
i.e., ranking by scores and breaking ties in favor of the lowest index.

\subsection{Feature selection methods}\label{thressvd}
Perhaps the most straightforward ranking method is based on thresholding the elements of the covariance or correlation matrix $S_{XY}=\Cov(\bX,\bY) \text{ or }\text{\bf Cor}(\bX, \bY)$:
i.e.,defining the score as
\begin{equation}\label{thresscore}s_{thres}(i) = ||S_{X_i Y}||_\infty\end{equation}
where $S_{X_i Y}$ is the $i$th row of $S_{XY}$.

We also consider a ranking method which uses the singular-value decomposition of the cross-correlation or covariance matrix.
Recall that the singular-value decomposition $S_{AB} = UDV^T$ is the unique matrix decomposition
in which $U$ and $V$ are semiorthogonal, and $D$ is diagonal with nonnegative entries in descending order.
The first left singular vector $u_1$ is the first column of $U$, and we define the score based on the absolute values
of the components of $u_1$:
\begin{equation}\label{svdscore}s_{SVD}(i) = |u_{1i}|.\end{equation}
It is known that the left singular vector $u_1$
satisfies the criterion
\begin{equation}\label{criterion}u_1 = \argmax_u Cov(\hbX u, \hbY v) \subjectto ||u||_2 = 1, ||v||_2 =1.\end{equation}
In comparison, the classical technique of canonical correlation analysis \cite{ANDERSON} finds $u, v$ which maximize $Cor(\hbX u, \hbY v)$.
However, the fact that canonical correlation analysis depends on inverting the inter-class sample covariance matrices $S_X, S_Y$ limits its applicability to data with small sample sizes.
Meanwhile, the \emph{sparse canonical correlation analysis} algorithm proposed by Witten\cite{WITTEN2}
proceeds by substituting $I_p$ and $I_q$ for $S_X, S_Y$, but this can be easily seen
to lead to an equivalent criterion to \eqref{criterion}.
However, Witten's algorithm allows for automatic inference of the number of significant features
through the use of an additional $\ell_1$ penalty to \eqref{criterion},
in contrast to our framework, in which we simply rank the features and leave to the user
the decision of how many features to keep.
For instance, in \S\ref{qvalues} we demonstrate the use of permutation-null derived false discovery rates
for determining how many genes to report.

At $n=2$, due to the fact that the sample covariance matrix is rank 1, both the thresholding and SVD methods
necesarily produce the same ranking.
However, for $n > 2$, the rankings can differ.

\subsection{Asymptotics}

As our ultimate goal is to obtain a general understanding of optimal feature selection under our model,
analytical results for the performance of all feature selection methods are indispensable.
Analagous results have been obtained for sparse PCA \cite{SPCA} using asymptotics
for $n \to \infty$ and also for the joint limit $n \to \infty, p \to \infty$.
For the multivariate feature selection problem, a variety of asymptotic limits can be considered:
by increasing the sample size to infinity while also changing the
number of correlated features, number of extraneous features, number of correlated or extraneous variates,
or a number of combination of these.
However, we find it most convenient to consider a limit in which the matrix $\Sigma_{XY}$ is shrunk to zero
as the sample size increases.

While we expect that the sample correlation matrix will be used more often than the sample covariance matrix
in applications, the intractable distribution of the sample correlation matrix \cite{SRIVASTAVA} leads us to consider
only the case in which $S$ is sample covariance matrix.
Then under our model,
it is possible to obtain asymptotic independence of the entries of scaled sample cross-covariance matrix $\sqrt{n}S_{XY}$
by letting $n \to \infty$ while simultaneously allowing the covariance matrix $\Sigma$
to change depending on the sample size.
This result is stated below.
\newline

\begin{theorem}\label{asym}
Let $\Omega$ be a $p \times q$ real matrix.
Define
\begin{equation}
\Sigma(n)=\begin{pmatrix}I_p & (\sqrt{n_0}/\sqrt{n})\Omega \\
(\sqrt{n_0}/\sqrt{n})\Omega^T & I_q
\end{pmatrix}.
\end{equation}
\begin{equation}S(n) \sim Wishart(n,\frac{1}{n}\Sigma(n))\end{equation}
and
let $S_{XY}(n)$ be the submatrix formed by the first $p$ rows and
the last $q$ columns of $S(n)$.
Then as $n\to \infty$,
$vec(\sqrt{n}S_{XY}(n))$ converges in distribution to $N(vec(\sqrt{n_0}\Omega), I_p \otimes I_q)$.\end{theorem}

\begin{proof}
Let $\sqrt{n}S(n) = (s_{ij})$, $\Sigma(n)=(\sigma_{ij})$, then
note that $Cov(s_{ij},s_{kl}) = \sigma_{ik}\sigma_{jl} + \sigma_{il}\sigma_{jk}$ (see \cite{MUIRHEAD}, p. 90).
Recall that $\sqrt{n}S_{XY}$ consists of the elements $s_{ij}$ where $i \leq p$ and $j > p$.
Thus $Cov(s_{ij},s_{kl})$ can be calculated by:
\begin{description}
\item[Case 1.] $i = k, j=l$.  Then $\sigma_{ik}=\sigma_{jl}=1$,
while $\sigma_{il}=\sigma_{jk}$.
Thus, $Var(s_{ij}) = 1 + \sigma_{il}^2.$
But for $i \leq p$ and $j > p$, $\lim_{n\to \infty} \sigma_{il} \to 0$ so
\begin{equation}\lim_{n\to\infty} Var(s_{ij}) \to 1.\end{equation}
\item[Case 2.] $i \neq k$ or $j\neq l$.
If $i\neq k$,  then $\sigma_{ik}=0$, since it lies off the diagonal of $\Sigma_X= I_p$.
Similarly, if $j \neq l$, then $\sigma_{jl}=0$ since it lies off the diagonal of $\Sigma_Y=I_q$.
Thus, $Cov(s_{ij},s_{kl})$ vanishes asymptotically.
\end{description}
The result then follows from applying the multivariate central limit theorem.\end{proof}

One can easily see that the matrix $\Sigma(n)$ as defined above is positive semidefinite for $n \geq n_0$ (\cite{BHATIA}).
Note that while $\sqrt{n}S_{XY}(n)$ converges to a distribution, the full matrix $\sqrt{n}S(n)$ fails to converge in distribution since its diagonal elements tend to infinity.

Now note that thresholding and SVD methods have the following expressions for the
1-0 loss when applied to matrix $T=S_{XY}$:
\begin{equation}\label{lthres}
L_{thres}(T) = I(\argmax_i ||(T^T)_i||_\infty\leq p_t)\end{equation}
and
\begin{equation}
L_{SVD}(T) = I(\argmax_i|(\argmax_u ||u^T T||)_i|) \leq p_t\end{equation}

Therefore we can compute the asymptotic approximation for the risk of the thresholding method as follows:
\newline

\begin{proposition}
Let $n_0$, $\Omega$, $\tilde{S}_{XY}$ be defined as in Theorem \ref{asym} and let $L_{thres}$ be defined as in \eqref{L_{thres}}.
Then,
\begin{align}
\lim_{n \to \infty}&\mathbb{E}(L_{thres}(\tilde{S}_{XY})) \\
= &\int_{x=0}^\infty (p_- q) F_{\chi^2_1}(x)^{p_- q - 1} f_{\chi^2_1} \bigg( 1- \prod_{i = 1}^{p_+} \prod_{j=1}^q F_{\chi^2_1,\omega_{ij}^2}(x)\bigg) dx
\end{align}
where $F_{\chi^2_1, d}(x)$ is the cdf of the noncentral chi-squared distribution with 1 degree of freedom and
noncentrality parameter $d \geq 0$, $F_{\chi^2_1}(x) = F_{\chi^2_1,0}(x)$ and $f_{\chi^2_1} = -\frac{d}{dx} F_{\chi^2, 1}(x)$.
\end{proposition}

\begin{proof}
Observe that when $Z \sim N(\mu, I_m)$
\begin{equation}Pr[||Z||_\infty < x] = \prod_{i=1}^m Pr[Z_i^2 < x^2],\end{equation}
meaning that 
\begin{equation}\mathbb{E}(L_{thres}(\tilde{S}_{XY}))=Pr[||vec(\tilde{S}_{X_t Y})||_\infty < ||vec(\tilde{S}_{X_u Y}||_\infty]\end{equation}
 can be calculated in terms of
noncentral chi-squared distributions.\end{proof}

In a similar way, bounds on $\mathbb{E}(L_{SVD}(\tilde{S}_{XY})$ can be obtained by comparing
the singular values of $S_{X_t Y}$ and $S_{X_u Y}$.
We also claim without proof that that such asymptotic approximations uniformly converge to the true risk function
for fixed $p$, $q$, and $0 < p_t < p$, as $n_0$ tends to infinity,
for any feature selection methods which follow the two conditions:
\begin{itemize}
\item \emph{Monotonicity with respect to sample-size}: \begin{equation}
\mathbb{E}(L(S_{XY}(n))) < \mathbb{E}(L(S_{XY}(n+1)))
\end{equation}
\item \emph{Monotonicity with respect to signal strength}: 
For all $\Omega \in \mathbb{R}_{p \times q}$ with $||\Omega|| \leq 1$, defining $\Sigma(\lambda) = \begin{pmatrix}I_p & \lambda\Omega \\ \lambda \Omega^T & I_q\end{pmatrix}$ for positive constant $\lambda <1$, $S(n, \lambda) \sim Wishart(n, \frac{1}{n}\Sigma(\lambda))$, and $S_{XY}(n,\lambda)$ as the submatrix comrpised of the first $p$ rows and $q$ columns of $S(\lambda)$,
 \begin{equation}
\mathbb{E}(L(S_{XY}(n, 1))) > \mathbb{E}(L(S_{XY}(n, \mu)))
\end{equation}
\end{itemize}
Additionally we claim that the thresholding method and the SVD method both satisfy these monotonicity conditions.
We postpone the technical justification of these claims for a forthcoming theoretical paper.

To determine the small-sample validity of the asymptotic approximation obtained above, we use
stochastic optimization applied to Monte Carlo simulations, as we discuss in the subsequent subsection.

\subsection{Computational Methods}

For simulation purposes we assume that $X, Y$ have a multivariate joint normal distribution with the covariance matrix \eqref{sigma}.
To reduce the size of the parameter space, we set $p_t=q_t$ and require that $\Sigma_{X_tY_t}$ be a
random matrix parameterized by a single parameter, $p_t$.
Specifically, we let
\begin{equation}\Sigma_{X_tY_t} = G_1 D G_2^T\end{equation}
where $G_1, G_2$ are independent random $p_t \times p_t$ orthogonal matrices and $D$
be a diagonal matrix with diagonal entries $d_1,\hdots,d_{p_t}$.

The resulting model consists of four parameters:
\begin{itemize}
\item $n$, the sample size
\item $p_t$, the number of correlated features and response variates
\item $p_u$, the number of extraneous features which are uncorrelated with the response
\item $q_u$, the number of components in the response vector uncorrelated with the explanatory variate
\end{itemize}

Under this model define the following functions of the parameters $n,p_t,p_u,q_u$:
\begin{eqnarray}
P_{thres}&=\mathbb{E}[1-L_{thres}(S_{XY})], \\
\tilde{P}_{thres}&=\mathbb{E}[1-L_{thres}(S_{XY})],\\ 
P_{SVD} &= \mathbb{E}[1-L_{SVD}(S_{XY})],\\
\tilde{P}_{SVD}&=\mathbb{E}[1-L_{SVD}(\tilde{S}_{XY})],\end{eqnarray} i.e. $P$ is the probability that the top-ranked feature is correlated to $Y$.

Using monte carlo simulations we can approximate $P_{thres}, P_{SVD}, \tilde{P}_{thres}, \tilde{P}_{SVD}$ by the following procedure
\begin{enumerate}
\item For monte carlo trials $i=1,\hdots,mc_{res}$ generate independent random orthogonal matrices \cite{STEWART} $G_1^{(i)},G_2^{(i)}$ and independent random diagonal matrices with uniform [0,1] entries $D^{(i)}$.
\item Form population cross-covariance matrices $\Sigma_{XY}^{(i)}$ by
$\Sigma_{XY}^{(i)} = G_1^{(i)} D^{(i)} G_2^{(i)T}$, and population covariance matrices $\Sigma^{(i)} = \begin{pmatrix}I_p & \Sigma_{XY}^{(i)}\\ \Sigma_{XY}^{(i)T} & I_q\end{pmatrix}$
\item Form sample cross-covariance matrices $S_{XY}^{(i)}$ by extracting the
first $p$ rows and last $q$ columns of a $Wishart(n-1,\frac{1}{n-1}\Sigma^{(i)})$ matrix.
\item Form asymptotic sample cross-covariance matrices $\tilde{S}_{XY}^{(i)}$ by
\begin{equation}
vec(\tilde{S}_{XY}^{(i)}) \sim N(vec(\Sigma_{XY}^{(i)}), \frac{1}{n-1}I_p \otimes I_q).\end{equation}
\item For each $S_{XY}^{(i)}, \tilde{S}_{XY}^{(i)}$ appy thresholding and SVD methods to obtain rankings
$\psi_{thres}^{(i)}, \psi_{SVD}^{(i)}, \tilde{\psi}_{thres}^{(i)}, \tilde{\psi}_{SVD}^{(i)}$ 
\item Compute approximate values of $P_{thres}, P_{SVD}, \tilde{P}_{thres}, \tilde{P}_{SVD}$
by 
\begin{eqnarray}
P_{thres} &= \frac{\sum_{i=1}^{mc_{res}} I(\psi_{thres}^{(i)}(1) \leq p_t)}{mc_{res}}\\
P_{SVD} &= \frac{\sum_{i=1}^{mc_{res}} I(\psi_{SVD}^{(i)}(1) \leq p_t)}{mc_{res}}\\
\tilde{P}_{thres} &= \frac{\sum_{i=1}^{mc_{res}} I(\tilde{\psi}_{thres}^{(i)}(1) \leq p_t)}{mc_{res}}\\
\tilde{P}_{SVD} &= \frac{\sum_{i=1}^{mc_{res}} I(\tilde{\psi}_{SVD}^{(i)}(1) \leq p_t)}{mc_{res}}
\end{eqnarray}
\end{enumerate}

In this paper we use stochastic search techniques to calculate approximate bounds on
\begin{equation}
\max_{p_t,p_u,q_u} |P_{thres} - \tilde{P}_{thres}|
\end{equation}
\begin{equation}
\max_{p_t,p_u,q_u} |P_{SVD} - \tilde{P}_{SVD}|
\end{equation}
for fixed $n$.

The problem of optimizing $P_{thres/SVD}-\tilde{P}_{thres/SVD}$ over a three-dimensional discrete parameter space $p_t, p_u, q_u$
can be handled via two different approaches\cite{SPALL}:
\begin{itemize}
\item Optimizing over a fixed grid of points $(p_t,p_u,q_u)$ using sequential testing methods.
\item Using a stochastic analogue of gradient descent, \emph{stochastic approximation}
\end{itemize}
However, in order to take advantage of our massively parallel computing setup, we develop a population-based optimization technique which
combines aspects of both approaches.
The proposed algorithm is outlined below:\newline

\noindent\emph{Algorithm 1}
\begin{enumerate}
\item Given a random variable $X|\theta$ over a parameter space $\theta$, we wish to find \[\argmax_\theta \mathbb{E}[X|\theta]\]
\item Starting with a grid of parameter values $\theta_1,\hdots,\theta_{k_0}$, compute empirical means
$\bar{X}|\theta_1,\hdots, \bar{X}|\theta_{k_0}$ using $mc_{res}$ repeated measurements at each parameter value
\item At the $t$th step, let $\theta^{t}_1, \hdots, \theta^{t}_m$ be the $m$ parameter values with the largest empirical means $\bar{X}$ among $\theta_1,\hdots,\theta_{k_{t-1}}$
\item Update $\bar{X}|\theta^{t}_1,\hdots,\bar{X}|\theta^{t}_m$ with $mc_{res}$ additional measurements of $X$ at each parameter value
\item Generate $\theta_{k_{t-1}+1}, \hdots, \theta_{k_t}$ by randomly perturbing $\theta^t_1,\hdots,\theta^t_m$, and 
compute empirical means $\bar{X}|\theta_{k_{t-1}+1},\hdots, \bar{X}|\theta_{k_t}$ using $mc_{res}$ repeated measurements at each parameter value.
\item Repeat until step $t_{final}$.
\end{enumerate}

Our optimization results are discussed in \S\ref{optres}.

\subsection{Computational Results}\label{optres}
Table 1 provides the results obtained for $n=2$ and $n=6$.  Note that
standard errors for all probabilities are less than $0.002$.
In both cases we run Algorithm 1 for $t_{final}=5$ steps, using $mc_{res}=5000$, $k_0 = 500$,
with $\theta_1,\hdots,\theta_{k_0}$ being grid points over $p_t = \{2,3,4,5,6\}, p_u =\{1,6,\hdots,41,46\}, q_u=\{1,6,\hdots,41,46\}$, $m=10$, and $k_t-k_{t-1}=100$, with $\theta_{k_{t-1}+1}, \hdots, \theta_{k_t}$
being generated by creating 10 perturbed copies of $\theta^{t-1}_1,\hdots,\theta^{t-1}_{10}$
with additive perturbations $(\delta_1,\delta_2,\delta_3)$ where $\delta_1$ is uniformly distributed
over $\{-1,0,1\}$ and $\delta_2, \delta_3$ independently and uniformly distributed over $\{-3,-2,\hdots,2,3\}$.

\noindent\begin{table}\begin{tabular}{ c|ccccc }
n = 2 & & & & & \\
  $\max{P_{thres}-\tilde{P}_{thres}}$ & $p_t$ & $p$ & $q$ & $P_{thres}$ & $\tilde{P}_{thres}$\\
\hline
  0.00 & 5 & 40 & 40 & 0.12 & 0.12\\
\\[-3ex]
  $\max{\tilde{P}_{thres}-P_{thres}}$ & $p_t$ & $p$ & $q$ & $P_{thres}$ & $\tilde{P}_{thres}$\\
\hline
0.06 & 2 & 7 & 2 & 0.29 & 0.34\\
\\[-3ex]
  $\max{P_{SVD}-\tilde{P}_{SVD}}$ & $p_t$ & $p$ & $q$ & $P_{SVD}$ & $\tilde{P}_{SVD}$\\
\hline
  0.00 & 2 & 53 & 55 & 0.03 & 0.03\\
\\[-3ex]
  $\max{\tilde{P}_{SVD}-P_{SVD}}$ & $p_t$ & $p$ & $q$ & $P_{SVD}$ & $\tilde{P}_{SVD}$\\
\hline
0.06 & 2 & 7 & 2 & 0.29 & 0.34\\
\\[-3ex]
n = 6 & & & & & \\
  $\max{P_{thres}-\tilde{P}_{thres}}$ & $p_t$ & $p$ & $q$ & $P_{thres}$ & $\tilde{P}_{thres}$\\
\hline
  0.01 & 2 & 53 & 53 & 0.05 & 0.04\\
\\[-3ex]
  $\max{\tilde{P}_{thres}-P_{thres}}$ & $p_t$ & $p$ & $q$ & $P_{thres}$ & $\tilde{P}_{thres}$\\
\hline
0.05 & 2 & 5 & 6 & 0.52 & 0.57\\
\\[-3ex]
  $\max{P_{SVD}-\tilde{P}_{SVD}}$ & $p_t$ & $p$ & $q$ & $P_{SVD}$ & $\tilde{P}_{SVD}$\\
\hline
  0.01 & 5 & 42 & 13 & 0.21 & 0.20\\
\\[-3ex]
  $\max{\tilde{P}_{SVD}-P_{SVD}}$ & $p_t$ & $p$ & $q$ & $P_{SVD}$ & $\tilde{P}_{SVD}$\\
\hline
0.05 & 2 & 5 & 6 & 0.52 & 0.57\\
\end{tabular} 
\caption{Stochastic optimization results for $n=2$ and $n=6$}
\end{table}

Note from Table 1 that the maximum discrepancy between the asymptotic result and the true small-sample
value decreases from $n=2$ to $n=6$ as might be expected.
However, these results are far from exhaustive, and it remains to perform the optimization for larger values of $n$
to confirm the apparent small-sample accuracy of the asymptotic approximation.

\section{Application.}\label{application}

\subsection{Summary}

In this section we apply the thresholding method and SVD method to select genes from a recent
microarray-metagenomics dataset (\S\ref{dataset}).  For each method we obtain a global permutation null distribution to determine false discovery rates for the corresponding ranked list of genes (\S\ref{qvalues}).

We use these q-values as a basis to determine which of the resulting rankings to use and to select how many genes to report from that ranked lists (\S\ref{results}).
The strongest results are obtained from applying the SVD method to the formula-fed data,
which accords with our simulation results indicating the relative strength of SVD for low sample sizes
and with previous observations of the relative homogeneity of the formula-fed data.
Based on a q-value cutoff of 0.15 we end up reporting ten genes:
MMD, PPP3CA, ALOX5, PAFAH2, C1QTNF6, MSRB3, VTN, ACVR1B, WASL, and MET.
To investigate the validity of the resulting q-values, we check our results against
rankings of genes from the thresholding and SVD methods combined with
an alternative permutation null.
We observe that although higher q-values result from the local null, the rankings of genes
resulting from SVD applied to the formula-fed data with the global null and the local nulls have high overlap.
In particular, PPP3CA and ALOX5 are top-ranked genes in both procedures.
We then apply the SVD procedure to identify metabiome attributes associated
with the ten selected genes, but none of the metabiome attributes are found to be siginificantly associated with the
selected genes.

We discuss possible biological interpretations of these findings in \S\ref{discussion}.

For the purpose of determining the sample size needed for a follow-up study, in \S\ref{simulation} we find the simulated performance of
the thresholding and SVD method as the sample and $p_t$, the true number of correlated genes and metabiome features, are varied.
From these results it is clear that while the SVD method dominates the thresholding methods at low sample sizes,
the thresholding method rapidly improves in performance as sample size increases and as $p_t$, the number
of correlated genes, increases.
Yet even under the most favorable conditions it appears that a sample size of around 100 is required
for reliable feature selection under our model, for $p=600$ and $q=200$.

\subsection{Dataset}\label{dataset}
The data originates from an experiment to study the effect of breast-feeding versus formula-feeding on infant health.
Stool samples were collected from six breast-fed babies and six formula-fed babies, and gene expression levels
were obtained via microarray intensities of host mRNA fragments isolated from the stool sample, while bacterial microbiome subsystem profiles were obtained by aggregating the fragments detected by metagenomic pyrosequence according to the three-level MG-RAST annotation \cite{AZIZ}.

Previous analyses characterised differences between the gene expression levels of the
two treatment groups \cite{CHAPKIN} and multivariate relationships between the host expression levels
and microbiome attributes which were potentially induced by the differences between treament groups \cite{SCHWARTZ}.
The current study is motivated by the goal of identifying mutalistic relationships between the host and the intestinal microbiome on the
basis of the microarray-metagenomics expression data for each treatment group seperately.

\subsection{Preprocessing.}
As per the suggestions in \cite{SCHWARTZ}, we focus on the immunology-related genes,
producing a data matrix of 6 observations by 585 genes, $\bX_{raw}$.
We select the microbial attributes with read counts higher than 300,
resulting in a data matrix of 6 observations by 211
microbial feature hit counts for each treatment group, $\bY_{raw}$.
We apply loess normalization to the log-transforms of the raw intensities in $\bX_{raw}$ \cite{SCHWARTZ},
standardize the rows and columns of $\bX_{raw}$ and $\bY_{raw}$ to have mean 0 and variance 1 as described in \cite{EFRON}
to arrive at the processed matrices $\bX$.
The hit counts in $\bY_{raw}$ are converted to proportions by individuals, then log-transformed, then row and column standardized to produce $\bY$.

\subsection{Procedure}\label{qvalues}
We form $S_{XY}=\Cov(\bX,\bY)$ and apply the thresholding and SVD methods to rank
the genes in $X$.

We also obtain false discovery rates ($q$-values) for each method by using a global row-wise permutation null
ditribution and prior false positive rate $\pi_0=1$ \cite{EFRON}.

For each method we compute a separate p-value $p_{thres},p_{SVD}$ for each gene via 
a global permutation null distribution for the scores $s_{thres}, s_{SVD}$ of the individual genes by the following:
\begin{enumerate}
\item Let $s_{thres}(j)$ be the score of the $j$th gene according to thresholding, as from \eqref{thresscore}, and $s_{SVD}(j)$ be the score of the $j$th gene according to SVD, as from \eqref{svdscore}.
\item For repetitions $i=1,\hdots,mc_{res}$ with $mc_{res} = 1000$, form permuted data matrix $\bY^{(i)}$
by independently permuting each row of $\bY$.  Then form cross-correlation matrices $S_{XY}^{i}$ from $\text{\bf Cor}(\bX,\bY)$.
\item Compute scores $s_{thres}^{(i)}$ and $s_{SVD}^{(i)}$ from $S_{XY}^(i)$.
\item Compute the $p$-values of the $j$th gene according to thresholding and SVD as:
\begin{equation}
p_{thres}(j) = \sum_{i=1}^{mc_{res}}\sum_{k=1}^p \frac{I(s_{thres}(j) \leq s_{thres}^{(i)}(k))}{mc_{res}p}
\end{equation}
\begin{equation}
p_{SVD}(j) = \sum_{i=1}^{mc_{res}}\sum_{k=1}^p \frac{I(s_{SVD}(j) \leq s_{thres}^{(i)}(k))}{mc_{res}p}
\end{equation}
\end{enumerate}

Next, let $\tau_{thres}(j)$ be the ranking of the $j$th gene in ascending order of the $p_{thres}$-values,
and let $\tau_{SVD}(j)$ be the ranking of the $j$th gene in ascending order pf the $p_{SVD}$-values, with ties broken in favor of the lowest index.
Note that $\tau = \psi^{-1}$ when a global null distribution is used.
Compute the false discovery rates $q_{thres}(j)$ and $q_{svd}(j)$ as
\begin{equation}
q_{thres}(j) = p (p_{thres}(j))/\tau_{thres}(j)
\end{equation}
\begin{equation}
q_{SVD}(j) = mp (p_{SVD}(j))/\tau_{SVD}(j)
\end{equation}
where $m$ is a correction factor for dependence \cite{BENJAMINI},
\begin{equation}
m = \sum_{j=1}^p \frac{1}{j}
\end{equation}
which evaluates to $6.95$ for $p=585$.

For comparative purposes we compute alternate p-values $\dot{p}_{thres}, \dot{p}_{SVD}$ according to a local permutation null distribution.  Note that resulting ascending ranking of $\dot{p}_{thres}, \dot{p}_{SVD}$ may differ from $\psi_{thres}, \psi_{SVD}$ respectively, since each gene has a unique null distribution.
The procedure is as follows:
\begin{enumerate}
\item For repetitions $i=1,\hdots,mc_{res}$ with $mc_{res} = 1000$, form permuted data matrix $\bY^{(i)}$
by permuting the row labels of $\bY$.  Then form cross-correlation matrices $S_{XY}^{i}$ from $\text{\bf Cor}(\bX,\bY)$.
\item Compute scores $s_{thres}^{(i)}$ and $s_{SVD}^{(i)}$ from $S_{XY}^(i)$.
\item Compute the $p$-values of the $j$th gene according to thresholding and SVD as:
\begin{equation}
\dot{p}_{thres}(j) = \sum_{i=1}^{mc_{res}}\frac{I(s_{thres}(j) \leq s_{thres}^{(i)}(j))}{mc_{res}}
\end{equation}
\begin{equation}
\dot{p}_{SVD}(j) = \sum_{i=1}^{mc_{res}}\frac{I(s_{SVD}(j) \leq s_{thres}^{(j)}(k))}{mc_{res}}
\end{equation}
\end{enumerate}
From these p-values $\dot{p}$ we obtain alternate rankings $\dot{\tau}^{-1}$ for thresholding and SVD.
We discuss the rankings $\psi_{thres}$ and $\psi_{SVD},\ \dot{\tau}^{-1}_{SVD}$ for the formula-fed data
in \S\ref{results}.

\subsection{Results}\label{results}

Table 2 provides the top three genes identified by
thresholding and SVD applied to the breast-fed data along with
q-values obtained from the global permutation null (\S\ref{qvalues}),
and Table 3 provides the analagous results for the formula-fed data.
Note that q-values for the SVD method can exceed 1 due to the correction factor for dependence.

Note that while thresholding has comparable q-values for the breast-fed and formula-fed data,
the SVD method produces extremely weak q-values for the breast-fed data but
extremely strong q-values for the formula-fed data.
This discrepancy in performance may be due to the increased variability in the
gene expression levels for the breast-fed data, as observed in \cite{CHAPKIN}
through examination of the raw intensities of ``housekeeping genes'' for the formula-fed and
breast-fed data.
Furthermore, it is already clear from Tables 2 and 3 that SVD applied to the formula-fed data
has the strongest results overall.  Table 4 provides the entire list of genes
produced by the SVD method applied to the formula-fed data with a q-value less than 0.15.
It is also worth noting that PPP3CA and PAFAH2 are common to both the top 10 genes
for the thresholding and SVD method; what is not shown is that there are no other
commonalities to the top 10 genes list.

Table 4 provides the alternate p-values $\dot{p}_{SVD}$ computed for the SVD method applied
to the formula-fed data using the local permutation null described in \S\ref{qvalues}.
The rankings $\psi_{SVD}$ and $\dot{\tau}^{-1}_{SVD}$ have high overlap in the sense
that 7 of the top 10 genes in $\psi$ are also among the top 10 genes in $\dot{\tau}^{-1}_{SVD}$:
namely: MMD, PPP3CA, ALOX5,  PAFAH2, C1QTNF6, VTN, and ACVR1B.
In particular, PPP3CA nad ALOX5 are in the top 3 genes in both permutation nulls.

From these results we judge it appropriate to select the top ten genes resulting from SVD applied
to the formula-fed data for further analysis.

In order to identify the metabiome attributes most closely associated with these ten genes,
we let $\bX$ be the metabiome data and $\bY$ be the intensities for the ten selected genes,
and apply SVD-based feature selection.
The results are listed in Table 5.
The first column of Table 5 provides the name of first SEED hierachy of the microbial attribute, which is the broadest categorization
in the MG-RAST SEED annotation scheme.
The second column is the name of the MG-RAST subsystem annotation, the finest level of the
hierarchical SEED annotation scheme and the level chosen for data aggregation.
While the q-values are very weak, it is worth noting that two of the top five attributes
belong to the virulence category, since only 11 of the 211 microbial attributes belong to the virulence category.
While two of the top five attributes also belong to the carbohydrates category, this is less interesting since
a total of 42 out of 211 of the microbial attributes belong to the carbohydrates category.

\noindent\begin{table}
\caption{Breast-fed data: Global null results}
\begin{tabular}{ c||cc||cc }
\hline
\# & name & $q_{thres}$ & name & $q_{SVD}$\\\hline
1 & THBS2 & 0.38 & GBP1 & 3.78\\
2 & FYN & 0.28 & TNFAIP8L1 & 2.27\\
3 & CRNN & 0.39 & TYROBP & 1.86\\\hline
\end{tabular} 
\end{table}

\noindent\begin{table}
\caption{Formula-fed data: Global null results}
\begin{tabular}{ c||cc||cc }
\hline
\# & name & $q_{thres}$ & name & $q_{SVD}$\\\hline
1 & PPARA & 0.15 & MMD1 & 0.00\\
2 & PPP3CA & 0.86 & PPP3CA & 0.00\\
3 & SDC4 & 0.39 & ALOX5 & 0.00\\
4 & PAFAH2 & 0.70 & PAFAH2 & 0.00\\
\hline
\end{tabular} 
\end{table}

\noindent\begin{table}
\caption{Fomula-fed data: SVD results}
\begin{tabular}{ c|cc|cc}
\hline
name & $\tau$& $q$ & $\dot{\tau}$ & $\dot{p}$ \\\hline
MMD &1& 0.00 & 5 & 0.002 \\
PPP3CA &2& 0.00 & 1 & 0.000 \\
ALOX5 &3& 0.00 & 2 & 0.000\\
PAFAH2 &4& 0.00 & 6 & 0.004 \\
C1QTNF6 &5& 0.00 & 10 & 0.011 \\
MSRB3 &6& 0.00 & 11 & 0.011\\
VTN &7& 0.00 & 3 & 0.002\\
ACVR1B &8& 0.00 & 4 & 0.002 \\
WASL &9& 0.08 & 27 & 0.040\\
MET&10& 0.11 & 14 & 0.013 \\
\hline
\end{tabular} 
\end{table}

\noindent\begin{table}
\caption{Metabiome attributes associated with selected genes}
\begin{tabular}{ cc|cc}
\hline
SEED 1 & name & $\tau_{SVD}$ & $q_{SVD}$\\\hline
Carb. & Se.-glyox. cycle & 1 & 0.31\\
Phos. & Control. PHO & 2 & 0.31\\
Viru. & CoZnCd res. & 3 & 0.28\\
Carb. & Beta-Gl. met. & 4 & 0.43\\
Viru. & Res. fluoroq. & 5 & 0.47\\
\hline
\end{tabular} 
\end{table}

\subsection{Discussion}\label{discussion}
The results of our analysis suggest that the gene PPP3CA merits further investigation.
While we could not conclusively determine which of the metabiome attributes were associated with
PPP3CA, we have relatively high confidence that PPP3CA is correlated with the metabiome attributes since 
the gene is highly ranked by multiple methods.
The gene PPP3CA codes for the enzyme calcineurin, which generates a signal activating the gut immune system \cite{VINDEROLA}.  One of calcineurin's specific functions is to dephosphorylate NFAT transcription factors
to promote immune activation \cite{RODRIGUEZ}.

The genes ALOX5 and PAFAH2 were also selected by more than one feature selection method.
In addition, ALOX5 was also selected in a previous study on the combined formula-fed and breast-fed data \cite{SCHWARTZ}. 
The gene ALOX5 codes for arachidonate 5-lipoxygenase, which is involved in mucosal inflammatory responses \cite{CLARK}.

While the results of the metabiome attribute selection were much weaker than the results of the
feature selection for the genes, it is intriguing that two of the top five metabiome attributes
were virulence-related: namely, cobalt-zinc-cadmium resistance and resistance to fluoroquinolones.
Correlations between the immunity and defense-related host genes and the virulence attributes
would agree with the biological intuition that the host would react to pathogens in the instestine;
or that conversely, that pathogenic activity may increase as a result of inhibited host immunodeficiency.

\subsection{Simulation results}\label{simulation}
In Figure 1 we show simulated results for $P_{thres}$ and $P_{SVD}$ for $p=600,q=200$ and $p_t$ varying from 10 to 100, $n$ varying from 2 to 100.  The height of the dark grey bars is the $P_{thres}$ and the height of the light grey bars is $P_{SVD}$
from 0 to 1.  The axis with the rising slope is the axis for $p_t$, taking values $(100, 90, \hdots, 10)$ from left to right.
The axis with falling slope is for $n$ taking values from $(10,20,\hdots,100)$ from left to right.
We used $mc_{res}=40000$ monte carlo trials for each parameter value; thus the standard errors $< 0.025$ result in confidence bounds which are too small to be visible.

\begin{figure}
\caption{Simulated performance of thresholding (dk. grey) versus SVD (lt. grey) for $p=600,\ q=200$}
\noindent\includegraphics[trim=3cm 3cm 10cm 2cm, clip=true, totalheight=2in]{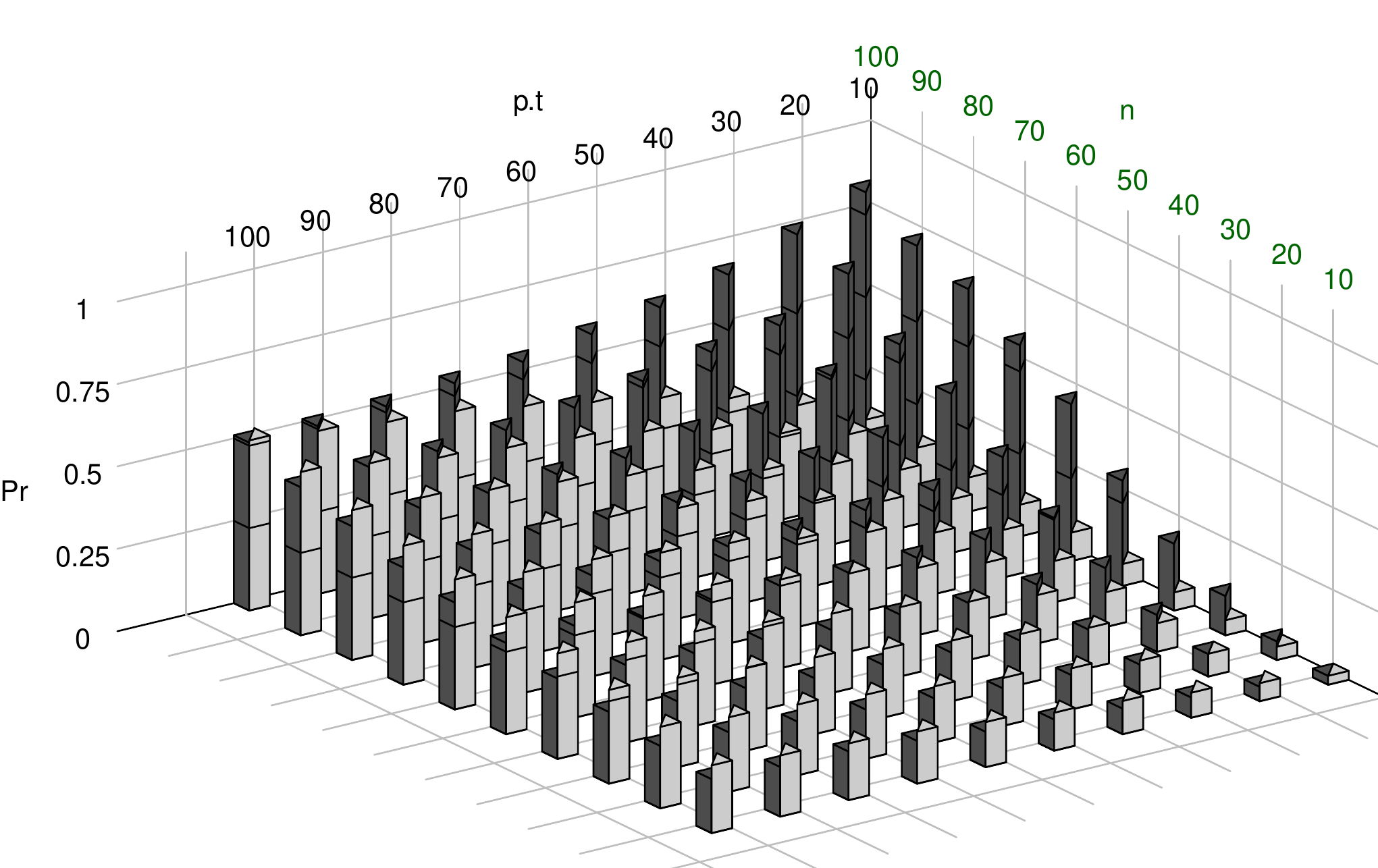}
\end{figure}

From the simulation we conclude that for plausible values of $p_t$, the
top ranked gene via thresholding (or SVD) is a false positive with probability exceeding 0.9.
However, SVD is indeed more effective than thresholding at $n=12$.  But as we can see, as the sample size increases,
the thresholding method rapidly climbs in relative effectiveness.
At $n=70$, the thresholding method has a higher probability of assigning the top ranking
to a correlated gene than the SVD method for $p_t < 50$.
At $n=100$, the thresholding method outperforms the SVD method for $p_t < 90$, which
encompasses most of the biologically plausible range for $p_t$.
Of note is the nonmonotonicity of the thresholding method with respect to $p_t$ for fixed $n,p,q$;
while both SVD and thresholding increase in effectiveness for increasing $p_t$ when $p_t$ is large,
thresholding experiences a dramatic increase in effectiveness for decreasing $p_t$ when $p_t$ is small.
Yet even under the best plausible conditions, with $p_t=10$ for thresholding, a minimum sample size of 100
is required for the top-ranked feature to be correlated to $Y$ even 80 percent of the time.

These significant discepancies in performance, however, would seem to indicate that neither the
thresholding method nor the SVD method can be claimed to be the ``optimal'' method, and that
there may exist an as-of-a yet undiscovered method which dominates both of these simple approaches.

\section{Impact and Significance}\label{conclusions}

Our simulation results succeed in providing a basic understanding of the differences between
the thresholding and SVD methods.
To our knowledge, such a comparative study of multivariate feature selection methods has never
appeared in the literature.
In addition, our model allows for the quantitative analysis of experimental design considerations.
Researchers desiring an understanding of an integrated biological system can use
the model proposed in the paper to determine the relative value of additional observations
versus measurements of additional biological features (depth versus breadth).
This approach provides an appreciation for the importance of having prior knowledge that can allow for elimination
of extraneous features or variates.

With respect to the original problem which motivated this work, our data analysis diagnostics
and simulation results demonstrate that singular value decomposition is an effective tool
for identifying correlations between genes and microbial attributes for small-sample microarray-metagenomics datasets.
Our data analysis of the infant microarray-metagenomics dataset indicate that the combination
of SVD-based feature selection with permutation-null-derived false discovery rates provides
a powerful framework for inferring host-microbiome interactions.

While we only scratch the surface of the multivariate feature selection problem in this paper,
by the same token, the tools we introduce can be employed in further studies on multivariate selection.
The asymptotic approximation for the sample cross-covariance matrix in our model can be used
for any feature selection method to be studied using our model.
We demonstrate how stochastic optimization can be used to evaluate the accuracy of
the asymptotic approximation. In addition, the same stochastic optimization techniques
can be used to compare the performances of two competing feature selection methods.

It would be interesting to compare the performance of group lasso,
sparse CCA and bayesian approaches to feature selection under our proposed model.
In particular, we expect our results on the SVD method to generalize to the performance of
sparse CCA feature selection methods due to the similarity between the algorithms (\S\ref{thressvd}).
Based on our simulation results,
we predict that the thresholding method also outperforms the sparse CCA method as
the sample size increases.

\subsection*{Acknowledgements}
We are indebted to the Texas A \& M Brazos Computing Cluster and Institute of Developmental
and Molecular Biology for access to computing resources, and to professors David B. Dahl, 
Mohsen Pourahmadi, and Joel Zinn for helpful discussions.
The infant microarray-metagenomics data was provided courtesy of Sharon M. Donovan,
of the Division of Nutritional Sciences, U. of Illinois, Urbana, IL.

\end{document}